%% file: main.tex
\documentclass[conference]{IEEEtran}
\IEEEoverridecommandlockouts
\usepackage{cite}
\usepackage{amsmath,amssymb,amsthm}
\usepackage{dutchcal}
\usepackage{bbm}
\usepackage{algpseudocode}
\usepackage{graphicx}
\usepackage{textcomp}
\usepackage{mathtools}
\usepackage{xcolor}
\usepackage{algorithm}
\usepackage{bm}

\def\BibTeX{{\rm B\kern-.05em{\sc i\kern-.025em b}\kern-.08em
    T\kern-.1667em\lower.7ex\hbox{E}\kern-.125emX}}
\usepackage{url}

\usepackage{tikz}
    \usetikzlibrary{shapes.arrows}
    
\usepackage{pgfplots}
\pgfplotsset{
compat=1.3,
legend style={font=\footnotesize, fill opacity=0.7,  draw opacity=1, text opacity=1, draw=white!15!black, legend cell align=left, align=left}, 
width=6cm, 
height=6cm,
yminorticks=false,
xminorticks=false,
title style={font=\small},
tick style={color=black},
tick label style={font=\small},
grid style={line width=.1pt, draw=gray!20},
major grid style={line width=.1pt,draw=gray!20},
}
\usepgfplotslibrary{colorbrewer}
\usepgfplotslibrary{groupplots}
\usepgfplotslibrary{fillbetween}
\usetikzlibrary{plotmarks}
\usepackage[]{fancyhdr}
\usetikzlibrary{patterns}
\pgfplotsset{every tick label/.append style={font=\footnotesize}}

\usepackage{tikz}
\usetikzlibrary{shapes, arrows, positioning}

\usepackage{subfig}
    
\usepackage{glossaries}
\input{acronyms.tex}

\newcommand{\E}[1]{\mathbb{E}\left[ #1 \right]} 
\newcommand{\mc}[1]{\mathcal{#1}}   
\newcommand{\mb}[1]{\mathbf{#1}}    
\DeclareMathOperator*{\argmax}{arg\,max}    
\newcommand{\set}[1]{\mathcal{#1}}

\algtext*{EndWhile}
\algtext*{EndIf}
\algtext*{EndFor}

\def \sfwidth{0.23\linewidth}
\def \sfheight{0.18\linewidth}

\definecolor{color4}{HTML}{FFD700}
\definecolor{color3}{HTML}{EA5F94}
\definecolor{color2}{HTML}{CD34B5}
\definecolor{color1}{HTML}{9D02D7}
\definecolor{color0}{HTML}{0000FF}
\definecolor{darkblue}{HTML}{00429D}
\definecolor{darkgreen}{HTML}{005c00}
\definecolor{gold}{HTML}{D4AF37}
\definecolor{darkred}{HTML}{910000}
\definecolor{darkslategray38}{RGB}{38,38,38}

\newtheorem{theorem}{Theorem}

\newtheorem{definition}{Definition}

\makeatletter
\newcommand*\titleheader[1]{\gdef\@titleheader{#1}}
\AtBeginDocument{%
  \let\st@red@title\@title
  \def\@title{%
    \bgroup\normalfont\large\centering\@titleheader\par\egroup
    \vskip1.5em\st@red@title}
}
\makeatother

\title{Push- and Pull-based Effective Communication \\ in Cyber-Physical Systems}
\titleheader{ \vspace{-1.25cm} \small This paper has been submitted to IEEE International Conference on Computer Communications. Copyright may change without notice. \vspace{-.75cm}}

\author{\IEEEauthorblockN{Pietro Talli, Federico Mason, Federico Chiariotti, Andrea Zanella}
\IEEEauthorblockA{Department of Information Engineering, University of Padova, Via G. Gradenigo 6/B, 35131, Padua, Italy \\
Emails: pietro.talli@phd.unipd.it, \{federico.mason, federico.chiariotti, andrea.zanella\}@unipd.it\vspace{-0.7cm}}
\thanks{This project was funded under the National Recovery and Resilience Plan (NRRP), Mission 4 Component 2 Investment 1.3 - Call for tender No. 341 (15 March 2022) of the Italian Ministry of University and Research, funded by the European Union NextGenerationEU Project. MUR grant number: PE\_0000001, Concession Decree 1549 (11/10/2022) adopted by the Italian Ministry of University and Research, CUP C93C22005250001, project title: RESearch and innovation on future Telecommunications systems and networks, to make Italy more smART (RESTART). F. Chiariotti's activities are funded by NRRP ``Young Researchers'' grant REDIAL (SoE0000009).}
}

\begin{document}

\maketitle

\begin{abstract}
    In \glspl{cps}, two groups of actors interact toward the maximization of system performance: the sensors, observing and disseminating the system state, and the actuators, performing physical decisions based on the received information. While it is generally assumed that sensors periodically transmit updates, returning the feedback signal only when necessary, and consequently adapting the physical decisions to the communication policy, can significantly improve the efficiency of the system.
    In particular, the choice between push-based communication, in which updates are initiated autonomously by the sensors, and pull-based communication, in which they are requested by the actuators, is a key design step.
    In this work, we propose an analytical model for optimizing push- and pull-based communication in \glspl{cps}, observing that the policy optimality coincides with \gls{voi} maximization.
    Our results also highlight that, despite providing a better optimal solution, implementable push-based communication strategies may underperform even in relatively simple scenarios.    
\end{abstract}

\begin{IEEEkeywords}
Effective communication, Value of Information, Pull-based communication, Cyber-Physical Systems
\end{IEEEkeywords}

\glsresetall

\section{Introduction}
\label{sec:intro}

The Industry 4.0 revolution has made \glspl{cps} a fundamental pillar for multiple applications, including area surveillance, data muling, factory automation, teleoperation, and autonomous mobility~\cite{zhang2019networked}.
A \gls{cps} is composed of sensors, which collect information about the system state, and actuators, that can physically modify it and control it.
In standard scenarios, it is assumed that sensor transmissions take place periodically, with the goal of minimizing the \gls{aoi} at the receiver side. Such an approach ensures that, on average, the actuators have the most recent information about the system conditions.

Recent works have shown that \gls{aoi}-based optimization may lead to sub-optimal performance, either because updates are unnecessary and may waste communication resources when the environment remains stable, or because an abrupt change in the system state may go unreported for a relatively long time, harming the control performance.
In more advanced systems, the communication policy is tailored to the \gls{voi} of the system, which implies that new transmissions are started whenever the environment evolution gets unpredictable at the receiver side.
A further improvement in overall performance is obtained by adapting the actuators' decisions to the network operations, thus jointly optimizing communication and control~\cite{mason2023multiagent}.
In the pursuit of this goal, it is possible to adopt the \gls{marl} paradigm~\cite{sutton2018reinforcement}, an extension of \gls{rl} to handle multi-agent scenarios.


In the case of \gls{marl} optimization, each actor (i.e., sensor or actuator) is controlled via a distinct \gls{rl} agent, which adapts its local actions towards the maximization of the overall performance.
The fact that multiple agents interact with and simultaneously adapt to the same environment leads to non-stationary training conditions, where estimating the optimal policy is not always guaranteed~\cite{busoniu2008comprehensive}.
Better solutions may be obtained by centralizing the control process in a single unit, which ensures perfect agent coordination at the cost of higher learning complexity and resource consumption.
However, centralized solutions require reliable, high-frequency updates, which may be not always feasible, especially in the case \gls{iot} scenarios, where energy and communication constraints make it challenging to achieve a full environment perception in real-time~\cite{lei2020deep}.

In the case of a distributed optimization of \glspl{cps}, the placement of the computational intelligence controlling the communication systems is a critical design choice. 
If \gls{rl} agents are installed at the sensors, it is possible to observe the current environment (or part of it) and update the actuators' knowledge as needed, in a \emph{push-based} fashion.
However, this solution may not be feasible if the sensors have limited computational and communication capabilities, as in the case of \gls{iot} networks.
The other option is to install the computational intelligence at the actuators, each of which is associated with \gls{rl} agent that has to both control the physical actions of the actuator itself and request new updates from the sensors.
In such a case, communication takes place in a \emph{pull-based} fashion~\cite{kosta2017age}, and decisions on when and whether to request an update are made without any information about the environment aside from the last state update and the elapsed time.

The pull-based problem was first presented in~\cite{monahan1982state} and recently addressed in~\cite{nam2021reinforcement}, where the authors consider two possible approaches for approximating the optimal policy: the first explicitly estimates the transition probability of the \gls{mdp} describing the system evolution, while the second jointly learns communication and control policies as part of a single model.
The push-based configuration has been investigated in several recent works on \gls{eff_com} systems~\cite{kim2019learning}, in which the actuators' policies must be adapted to the frequency of communication updates.
However, the joint optimization of \gls{eff_com} and control policies is still relatively unexplored, and, to the authors' knowledge, no studies have directly compared the push- and pull-based communication approaches and analyzed the interdependency between \gls{cps} optimization and the \gls{voi}. 

To overcome these limitations, this work presents an analytical model for optimizing and evaluating push- and pull-based strategies in \glspl{cps}.
We consider a simple \gls{cps} scenario with a single actuator, without local sensing capabilities, that is connected to a \gls{bs} via a constrained communication channel. 
Using the channel involves a significant cost but represents the only way for the actuator to observe the system state.
This design involves a trade-off between the minimization of the communication cost and the quality of control, which becomes less accurate as the transmission rate is reduced, in both the push- and pull-based versions.
Using this model, we analyze the advantages and drawbacks of each configuration, proving relevant results and showing that the push-based system, while having better performance at the optimum, is a PPAD-hard problem~\cite{deng2023complexity}.

\section{System Model}
\label{sec:system}

We consider a system model with a single Actuator that can perceive the environment only through the information provided by a \gls{bs}.
This latter has full access to the system state and, therefore, embodies the sensor network described in the introduction.
The goal of the model is to design and evaluate \gls{eff_com} strategies between the Actuator and the \gls{bs}.
In the rest of this section, we will first formalize the environment as an \gls{mdp}, and then characterize the possible solutions in the pull- and push-based communication scenarios. 


\subsection{Markov Decision Processes}
\label{subsec:rl}

We consider an \gls{mdp} defined by the tuple $\langle \set{S}, \set{A}, \mb{P}, r, \gamma \rangle$, where $\set{S}$ is the set of states, $\set{A}$ is the set of actions, $\mb{P}$ is the full transition matrix, $r$ is the reward function and $\gamma \in [0,1)$ is the discount factor.
Symbol $\mathbf{P}^a$ denotes the transition matrix associated with action $a \in \mathcal{A}$.
Each row of $\mathbf{P}^a \in \mathbb{R}^{|\set{S}|\times |\set{S}|}$ defines the distribution of the next state over $\set{S}$.
For it to be a valid distribution, we impose $P^a_{ss'}\geq 0 \ \forall a \in \mathcal{A}, (s, s')\in \mathcal{S}^2$ and $\sum_{s^\prime} P^a_{ss^\prime} = 1$.
The immediate reward for taking action $a$ in state $s$ and transitioning to state $s'$ is denoted by $r_{s,s'}^a$.

In our problem, the Actuator cannot directly observe the state but incurs a cost $c_t$ to obtain it from the \gls{bs}, which observes it at every step with no associated cost.
Hence, we can consider two possible scenarios:
in the pull-based communication scenario, the Actuator itself decides whether to request a state update from the \gls{bs}, which constitutes a passive actor in the system;
in the push-based communication scenario, the \gls{bs}  itself can decide whether to transmit the system state to the Actuator.
Notably, in the first case, we have a single agent installed at the Actuator, while, in the second case, we have two distinct agents that need to cooperate. 

In the simplest case, every sensing action has the same cost, and the communication action is $c_t \in \lbrace 0,1 \rbrace$.
When the state $s_t$ is not available to the Actuator, its \emph{a priori} belief distribution can be computed using the Markovian property of the system state evolution.
Knowing the last observed state $s_{t-k}$ and the vector $\mb{a}=\lbrace a_{t-k}, a_{t-k+1}, \ldots, a_{t-1}\rbrace$ containing the $k$ actions undertaken since then, as well as the transition matrix $\mb{P}$, we can determine the belief distribution $\bm{\theta}_{s_{t-k},k}^{\mb{a}}$:
\begin{equation}
 \bm{\theta}_{s_{t-k},k}^{\mb{a}}=\left(\prod_{\ell=t-k}^{t}\mb{P}^{a_{\ell}}\right)\mathbbm{1}_{s_{t-k}},\label{eq:belief}
\end{equation}
where $\mathbbm{1}_s$ is a one-hot column vector whose elements are all 0, except for the one corresponding to $s$, which is equal to 1.

To jointly optimize the control and communication policies, we need to solve the following problem: 
\begin{equation}
\label{eq:constrained_problem}
    \text{maximize } \mathbb{E} \left[ \sum_{t=0}^\infty \gamma^t r_t \right]\     \text{such that } \mathbb{E} \left[ \sum_{t=0}^\infty \gamma^t c_t \right] \leq C,
\end{equation}
where $C$ is the cumulative sampling cost that the system can tolerate and $r_t$ is the instantaneous reward at time $t$.
This problem can be seen as a constrained \gls{pomdp} where each action, or combination of actions by the \gls{bs} and Actuator, corresponds to the pair $(a_t, c_t)$.
By defining a fixed communication cost $\beta\in\mathbb{R}^+$ and using it as a dual parameter controlling the trade-off between the communication cost and the reward, we can reformulate \eqref{eq:constrained_problem} as an unconstrained \gls{pomdp}:
\begin{equation}
    \text{maximize } \mathbb{E}\left[ \sum_{t = 0}^\infty \gamma^t (r_t - \beta c_t) \right].
\end{equation}
In the following, we will denote values and functions related to the Actuator with the subscript $A$, while the \gls{bs} will be associated with the subscript $B$.

The pull-based communication scenario can then be modeled as a problem known as an \gls{acnomdp}~\cite{nam2021reinforcement}, in which the state $s_t$ is only available to the Actuator by means of a (costly) sensing action.
Specifically, the sensing action corresponds to a request to the \gls{bs}, which has perfect state information and can transmit it on demand.
If we fix the polling strategy, the resulting system can be treated as a \gls{pomdp} that can be solved optimally with point-based value iteration~\cite{pineau2003point} or other methods that consider the $\alpha$-vectors~\cite{smith2005point}.

The push-based problem is more common in the \gls{eff_com} literature and is often modeled as a Remote \gls{pomdp}~\cite{tung2021effective}.
In this case, communication is initiated by the \gls{bs}, which knows both the last transmission to the receiver $s_{t-k}$ and the current state $s_t$, along with the time since the last update $k$. 
Intuitively, the performance of a push-based system should improve with respect to the pull-based one, as the former can exploit information on the specific realization of the system state trajectory, rather than just on its statistics.
On the other hand, remote \glspl{pomdp} are multi-agent problems and may involve critical issues in terms of coordination between the \gls{bs} and Actuator, which will be explored in the following.

\setlength{\textfloatsep}{5pt}

\section{Analytical Solution}
\label{sec:method}

To solve the \gls{mdp}, we consider a model-based approach: any \gls{rl} agent interacting with the environment can compute optimal control and communication policies with complete knowledge of the reward function and the state transition model. This can also be accomplished with statistical learning~\cite{nam2021reinforcement}, i.e., by learning the transition probabilities of the underlying \gls{mdp} before solving the communication-constrained problem. 
This choice allows for an easy comparison of the policies, without any training issues, and the results can be applied directly to any properly trained agent.  

\subsection{Pull-based Communication}

To solve the pull-based communication problem optimally, we propose a modified \gls{pi} algorithm to learn the communication control policies jointly.
We recall that \gls{pi} always converges to the optimal solution in single-agent \glspl{mdp}, as the one considered in pull-based scenarios.

We formulate a new \gls{mdp} in which each time step is divided into two sub-steps: one for the control action and one for the communication action.
The state of this new \gls{mdp} is also expanded to the tuple $(s_{t-k},k)$, which is available to the agent.
The action space is then $\mathcal{A}$ in control sub-steps, and $\{0,1\}$ in communication sub-steps.
The agent then splits the policy into two parts:
\begin{itemize}
    \item $\pi_{A}(s_{t-k}, k)$, the control policy, which maps each state $(s_{t-k}, k) \in \mc{S} \times \mathbb{N}$ to an action $a \in \mc{A}$;
    \item $\Delta(s_{t-k})$, the communication policy, which maps each state $s_{t-k} \in \mc{S}$ to the time until the next update, $c \in \mathbb{N}^+$.
\end{itemize}

\begin{algorithm}[t]
\caption{Pull-Based Modified Policy Iteration}
\label{alg:pi}
\begin{algorithmic}[1]
\footnotesize

\Require $\mathbf{P},r,\beta$
\State Initialize $\mb{V}_{A}(s,k) \gets 0$, randomize $\bm{\pi}_{A}(s,k)$, $\bm{\Delta}$

\While {true}
    \For{$(s,k)\in \mathcal{S}\times\lbrace 0,...,T_{\max} \rbrace$}
        \State $V_{A}'(s,k)\gets$Update using~\eqref{eq:control_update}
    \EndFor
    \State $\mb{V}_{A}\gets \mb{V}_{A}'$\Comment{Value update step}

    \For{$(s,k)\in \mathcal{S}\times\lbrace 0,...,T_{\max} \rbrace$}
        \State $\pi_{A}'(s,k)\gets$Update using~\eqref{eq:control_improvement}
        \State $\Delta'(s)\gets$Update using~\eqref{eq:comm_improvement}
    \EndFor
    \If {$\bm{\pi}_{A}'=\bm{\pi}_{A}\wedge\bm{\Delta}'=\bm{\Delta}$}
        \State\Return $\bm{\pi}_{A}, \bm{\Delta}$ \Comment{Convergence}
    \Else
        \State $\bm{\pi}_{A},\bm{\Delta}\gets\bm{\pi}_{A}',\bm{\Delta}'$    \Comment{Policy improvement step}
    \EndIf
\EndWhile
\end{algorithmic}
\end{algorithm}

Since the policy $\bm{\pi}_{A}$ depends only on the last observed state, we can determine the vector $\mb{a}^{\bm{\pi}_{A}}_s$ containing the $k$ actions chosen by the Actuator after observing $s$.
Following the bootstrap principle, the modified \gls{pi} algorithm then updates the value of a control sub-state as follows:
\begin{equation}\label{eq:control_update}
\begin{aligned}
    V_{A}'(s,k) = \sum_{\mathclap{(s',s'')\in\mc{S}^2}} \theta^{\mb{a}^{\bm{\pi}_{A}}_s}_{s,k}(s')P^{\pi_{A}(s,k)}_{s',s''}\bigg[r_{s',s''}^{\pi_{A}(s,k)}\\
    +\gamma\left(\delta_{s,k} (V_{A}(s'',0)-\beta)+\left(1-\delta_{s,k}\right)V_{A}(s,k+1)\right)\bigg],\\
\end{aligned}
\end{equation}
where $\delta_{s,k}$ is a utility function equal to $1$ if $\Delta(s)=k+1$, i.e., if an update is requested after $k+1$ steps, and 0 otherwise.
We observe that the division into sub-steps does not affect the reward and the solution to the modified problem is also optimal for the original \gls{mdp}.
Therefore, we do not then need to compute the control sub-state values explicitly.

\begin{algorithm}[t]
\caption{Push-Based Alternate Policy Iteration}
\label{alg:itpi}
\begin{algorithmic}[1]
\footnotesize

\Require $\mathbf{P},r,\beta, \bm{\pi}_{B}^{(0)}$
\State Initialize $\bm{\pi}_{B}\gets\bm{\pi}_{B}^{(0)}$, randomize $\bm{\pi}_{A}$
\While {true}
    \State $\bm{\pi}_{A}'\gets$\Call{ControlPolicyIteration}{$\mathbf{P},r,\beta, \bm{\pi}_{B}$}
    \State $\bm{\pi}_{B}'\gets$\Call{CommunicationPolicyIteration}{$\mathbf{P},r,\beta, \bm{\pi}_{A}'$}
    \If {$\bm{\pi}_{A}'=\bm{\pi}_{A}\wedge\bm{\pi}_{B}'=\bm{\pi}_{B}$}
        \State\Return $\bm{\pi}_{A},\bm{\pi}_{B}$ \Comment{Convergence}
    \Else
        \State $\bm{\pi}_{A},\bm{\pi}_{B}\gets\bm{\pi}'_{A},\bm{\pi}'_{B}$
    \EndIf
\EndWhile
\end{algorithmic}
\end{algorithm}

To perform the policy improvement step, we modified the standard update rule by using a reward that considers the entire evolution of the Markov chain until the state is sampled again.
If we consider a state sequence $\mb{s}'$, which begins $k$ steps after state $s$ is observed, the belief $\Theta^{\bm{\pi}_{A}}_{s,k}(\mb{s}')$ follows from~\eqref{eq:belief}: 
\begin{equation}\label{eq:belief_policy}
 \Theta^{\bm{\pi}_{A}}_{s,k}(\mb{s}')=\theta_{s,k}^{\mb{a}_s^{\pi_A}}(s'(1))\prod_{\ell=1}^{\mc{L}(\mb{s}')-1}P^{\pi_{A}(s,\ell+k-1)}_{s'(\ell),s'(\ell+1)},
\end{equation}
where $\mc{L}(\cdot)$ is a function whose outcome is the dimensionality of the input. Hence, the control policy is improved as follows:
\begin{equation}\label{eq:control_improvement}
\begin{aligned}
\pi_{A}'(s,k) = &\argmax_{a\in\mc{A}}\quad \sum_{\mathclap{\mb{s}'\in \mc{S}^{\Delta(s)-k+1}}} \Theta^{\bm{\pi}_{A}}_{s,k}(\mb{s}')\Bigg[\sum_{{\ell=1}}^{{\Delta(s)-k}}\gamma^{\ell-1}
r_{s'(\ell),s'(\ell+1)}^{a_s^{\bm{\pi}_{A}}(\ell+k-1)} \\&+\gamma^{\Delta(s)-k} \left(V_{A}(s'(\Delta(s)-k+1),0) -\beta\right)\Bigg].
\end{aligned}
\end{equation}
Instead, the communication policy is updated as:
\begin{equation}\label{eq:comm_improvement}
\begin{aligned}
    \Delta'(s) = &\argmax_{n\in\mathbb{N}^+}\sum_{\mathclap{\mb{s}'\in\mc{S}^n}}\Theta^{\bm{\pi}_{A}}_{s,1}(\mb{s}')\Bigg[\sum_{\ell=1}^{n-1}\gamma^{\ell}r_{s'(\ell),s'(\ell+1)}^{\pi_{A}(s,\ell)}\\
    &+\gamma^n(V_{A}(s'(n),0)-\beta)+r_{s,s'(1)}^{\pi_{A}(s,0)}\Bigg].
\end{aligned}
\end{equation}
Using the whole sequence of intermediate steps in the update equations is essential to take into account that the states $(s,k) \ \forall \, k>0$ are non-Markovian, as the belief distribution depends on the action sequence taken since the last communication.
The pseudocode for the \gls{pi} scheme for pull-based communication is given in Alg.~\ref{alg:pi}.

\begin{figure*}
    \centering
    \subfloat[Average reward (pull-based, focused).\label{fig:r_pull_focused}]{\input{results/colormaps/reward_matrix_pull_sparse_reward}}\hfill
    \subfloat[Update frequency (pull-based, focused).\label{fig:c_pull_focused}]{\input{results/colormaps/cost_matrix_pull_sparse_reward}}\hfill
    \subfloat[Average reward (pull-based, spread).\label{fig:r_pull_spread}]{\input{results/colormaps/reward_matrix_pull_spread_reward}}\hfill
    \subfloat[Update frequency (pull-based, spread).\label{fig:c_pull_spread}]{\input{results/colormaps/cost_matrix_pull_spread_reward}}\\
    \vspace{-7pt}
    \centering
    \subfloat[Average reward (push-based, focused).\label{fig:r_push_focused}]{\input{results/colormaps/reward_matrix_push_sparse_reward}}\hfill
    \subfloat[Update frequency (push-based, focused).\label{fig:c_push_focused}]{\input{results/colormaps/cost_matrix_push_sparse_reward}}\hfill
    \subfloat[Average reward (push-based, spread).\label{fig:r_push_spread}]{\input{results/colormaps/reward_matrix_push_spread_reward}}\hfill
    \subfloat[Update frequency (push-based, spread).\label{fig:c_push_spread}]{\input{results/colormaps/cost_matrix_push_spread_reward}}
    \caption{Figure of the Reward and Communication Cost for different densities and values of $\beta$.}
    \label{fig:perf}
\end{figure*}
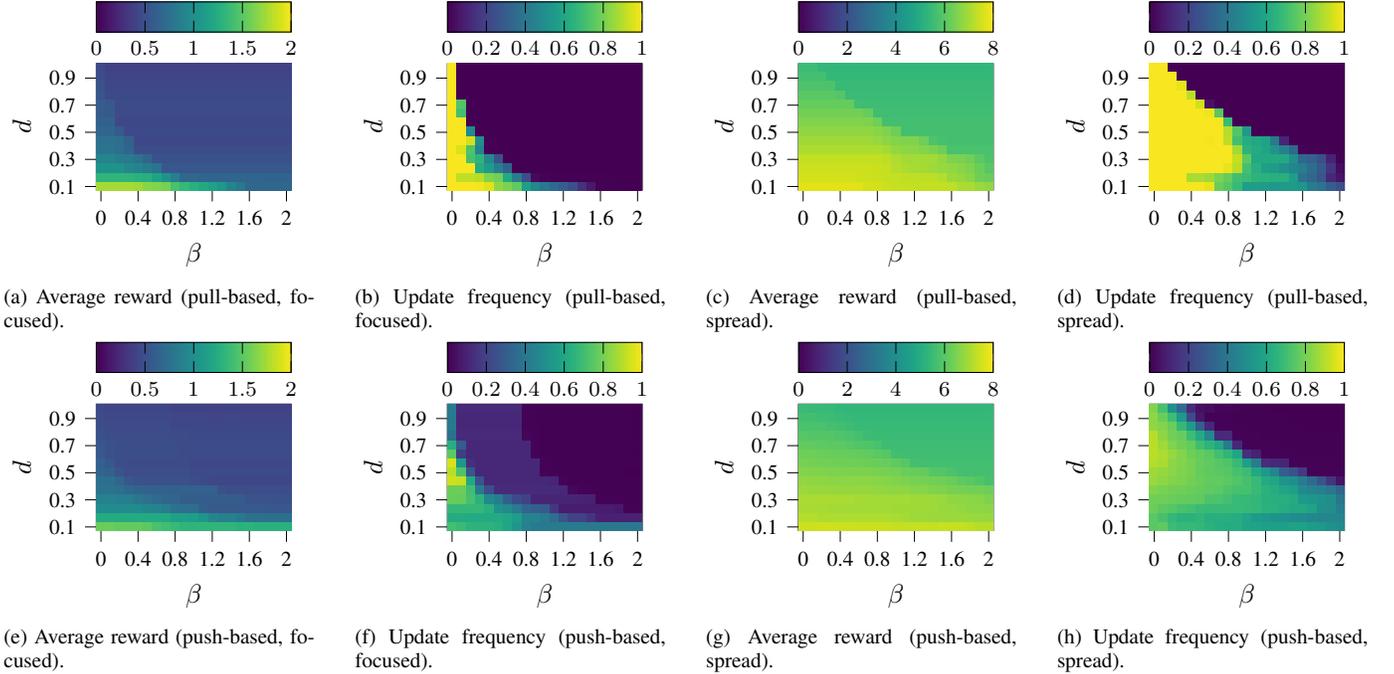

\subsection{Push-based Communication}

In the case of push-based communication, we propose a simple strategy which we call \gls{api}:
starting from a default policy for the \gls{bs}, we first optimize the Actuator's policy and then optimize the \gls{bs}, considering the Actuator's actions as fixed.
Since one agent's policy is fixed, the other agent sees a Markovian environment, and standard \gls{pi} can be applied in each step.
The procedure is then repeated until both the policies converge, i.e., each agent follows the optimal policy with respect to the other.
As we will prove in Sec.~\ref{sec:voi}, this strategy always converges in a finite number of steps.
The pseudocode for \gls{api} is given in Alg.~\ref{alg:itpi}.

\section{Numerical Evaluation}
\label{sec:numerical_evaluation}

In this section, we evaluate the two methods proposed on randomly generated \glspl{mdp} that follow a common structure for a fair and easy comparison.
We consider a system with $|\set{S}|=30$ states and $|\set{A}|=4$ actions. The reward depends only on the reached state and is defined as
\begin{equation}
    r_{s,s'}^a= e^{-\alpha|s'-s_0|},\  \forall s \in \set{S}, a \in \set{A},
\end{equation}
where $s_0 \in \set{S}$ is a target state, and $\alpha$ is a decay parameter determining how much the reward is concentrated around $s_0$.
We consider two possible reward configurations:
\begin{itemize}
    \item \textbf{Focused reward}: $\alpha=10$, the only state that gives a significant reward is $s_0$.
    \item \textbf{Spread reward}: $\alpha=0.01$, the reward is spread among the states near $s_0$.
\end{itemize}

In both configurations, the \gls{mdp} transition matrix is generated as follows: we start from a deterministic transition matrix for every action and then progressively increase the number of non-zero elements by enabling transitions to neighboring states with a certain probability.
As more transitions are added, the transition matrix becomes denser, i.e., the number $d$ of non-zero elements divided by the total number of entries of the matrix increases. The density of the deterministic matrix is simply $|\set{S}|^{-1}$.
If the initial deterministic matrix $\mb{P}^a$ has a transition from $s$ to $s'$, i.e., $P_{s,s'}^a=1$, the distribution at density $d$ is: 
\begin{equation}
    P_{s,s''}^a = \frac{4(d |\set{S}|-2|s'-s''|)}{(d |\set{S}| +1 )^2}, \ \forall s'' \ \text{s.t.} \ |s'-s''| < \frac{d |\set{S}|}{2}.
\end{equation}
According to this approach, the probability of transitioning to neighbor states decreases linearly with the distance from the original transition state (the distribution has a triangular shape). Higher-density \glspl{mdp} are then simply more unpredictable versions of the same initial model.

To investigate our model in different scenarios, we repeatedly increase the number of transitions for each state by adding two connections at a time, obtaining a total of 15 \glspl{mdp} with increasing densities.
We also consider different communication costs $\beta \in \{0,0.1,\ldots,2\}$.
We test each combination of $d$ and $\beta$ with both pull-based and push-based policies, using the same matrices for the push- and pull-based systems. The code for the numerical evaluations is available to replicate the results\footnote{\url{https://www.github.com/pietro-talli/Age-Value-CC}}.

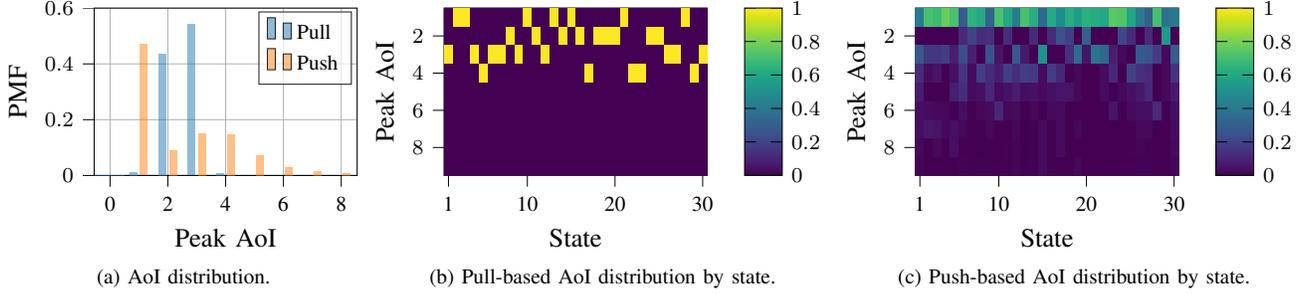
\begin{figure*}
    \centering
    \subfloat[AoI distribution.\label{fig:aoi_dist}]{\input{results/AOI/all_states}}
    \subfloat[Pull-based AoI distribution by state.\label{fig:pull_aoi}]{\input{results/AOI/heatmap_pull}}
    \subfloat[Push-based AoI distribution by state.\label{fig:push_aoi}]{\input{results/AOI/heatmap_push}}
    \caption{Peak AoI distribution for $d=0.1$, $\beta=1$ and sparse reward.}
    \label{fig:PeakAOI}
\end{figure*}

Fig.~\ref{fig:perf} shows the average obtained reward of the physical process and the frequency of the updates between the \gls{bs} and the Actuator. First, we can note that the push-based system always obtains a higher overall reward than the pull-based one, as would be expected: the \gls{bs}, acting with full knowledge of the state, can determine whether an update has a high \gls{voi}, i.e., whether it meaningfully affects the Actuator.

Fig.~\ref{fig:c_pull_focused} and Fig.~\ref{fig:c_pull_spread} also show that the pull-based system only has $2$ working points for higher values of $d$ (i.e., more unpredictable transitions): either the agent requests an update at each step, or it never does.
This is because the high uncertainty in the belief distribution makes it hard to use it to predict the optimal action even after a single step.
This leads the Actuator to always request the new state information if $\beta$ is low, i.e., if the relative \gls{voi} of the update increases, or never to request updates if $\beta$ is high, moving blindly. 

As Fig.~\ref{fig:c_push_focused} and Fig.~\ref{fig:c_push_spread} clearly show, the push-based strategy can reduce the update frequency.
In this case, the \gls{bs} can check if the state of the system evolves in the predicted way, transmitting an update to the Actuator only when a low-probability transition occurs and its knowledge becomes obsolete. This is particularly beneficial if the transition model is more predictable, i.e., for lower values of $d$, and leads to $3$ working regimes: firstly, if $d$ is very high, the \gls{bs} communicates only if the reward in the next step changes significantly, letting the Actuator move blindly in all other cases. This behavior is particularly noticeable in the focused reward case when the \gls{bs} only communicates if the state can lead directly to the target.

On the other hand, the \gls{bs} will also communicate sporadically if the system is very predictable, i.e., for very low values of $d$, for the reasons we outlined above. In between these two extremes, the \gls{bs} will communicate more often, as the system is not so random that the belief distribution of the Actuator quickly becomes useless, but also not too predictable.
As this type of behavior requires the full knowledge of the state, it cannot be obtained in a pull-based scenario.

To better highlight the different behaviors between pull- and push-based communication, we look at Fig.~\ref{fig:PeakAOI}, which analyzes the Peak \gls{aoi} for a highly predictable \gls{mdp} with a sparse reward, with $d=0.1$ and $\beta=1$.
In particular, Fig.~\ref{fig:aoi_dist} shows the \gls{pmf} of the Peak \gls{aoi} for the two approaches: while the average update frequency is similar, the push-based system has a much higher variance than the pull-based one.
This is because the first approach triggers transmissions when needed, sometimes after a single step, sometimes after many, depending on the environment's evolution; instead, the pull-based approach only relies on the last observed state, and its decisions are always deterministic.
This is confirmed by Fig.~\ref{fig:pull_aoi}, which shows how, in the pull-based scenarios, the Peak \gls{aoi} distribution for each state is concentrated in a single point.
On the other hand, the push-based system may transmit after a long time even when starting from states that have a high expected \gls{voi} after a single step, as Fig.~\ref{fig:push_aoi} shows: while the pull-based system transmits after $1$ step when starting from state $2$, the push-based system has a non-zero probability of waiting $5$ or more steps before transmitting again.

\section{Age and Value of Information in Effective Communication}\label{sec:voi}
Concepts like \gls{aoi} and \gls{voi} are crucial to analyze the timeliness and relevance of system status updates.
In our scenario, we might consider an update strategy that balances the frequency of updates (which determines the average communication cost) with the reward obtained by the Actuator. Let $J(s)=\mathbb{E}\left[\sum_{t=0}^\infty \gamma^tr_t \mid s_0=s\right]$ be the long-term reward of the Actuator starting from state $s$ and $\beta$ be the communication cost.
We can then give a more general definition of performance in pull-based or scheduled systems as a function of the selected update periods $\bm{\Delta} \in (\mathbb{N}^+)^{|\mc{S}|}$:
\begin{equation}
\label{eq:tot_r}
R_{\beta}(\bm{\Delta})=\!\sum_{\mathclap{s \in \mc{S}}} \phi(s)\!\left(\! J(s|\bm{\Delta})\!-\!\E{\sum_{n=0}^\infty\beta\gamma^{n\Delta(s')}P(s'|s,\bm{\Delta})\!}\right)\!,
\end{equation}
where $\phi(s)$ is the stationary state distribution induced by the control policy and the update periods $\bm{\Delta}$. 

We introduce the concept of Pareto dominance~\cite{pareto1919manuale} to compare the performance of complex schemes with multiple objectives.
\begin{definition}
an $n$-dimensional tuple $\bm{\eta}=(\eta_1,\ldots,\eta_n)$ \emph{Pareto dominates} $\bm{\eta}'$ (which we denote as $\bm{\eta}\succeq\bm{\eta}'$) if and only if each element of $\bm{\eta}$ is equal to or better than the corresponding element of $\bm{\eta}'$, i.e., $\bm{\eta}\succeq\bm{\eta}'\iff\eta_j\geq\eta_j'\ \forall  j$. Strict Pareto dominance makes the inequality strict for at least one parameter:
\begin{equation}
\bm{\eta}\succ\bm{\eta}'\iff\bm{\eta}\succeq\bm{\eta}'\wedge\exists i: \eta_i > \eta_i'.
\end{equation}    
\end{definition}
This concept can then be extended to multi-objective optimization, i.e., to schemes that may have different parameters and multiple performance metrics.
\begin{definition}
Let us consider two schemes $x$ and $y$, parameterized by a vector $\bm{\theta}_x$ and $\bm{\theta}_y$. Scheme $x$ Pareto dominates $y$ if and only if:
\begin{equation}
x\succeq y\iff\forall\bm{\theta}_y\exists\bm{\theta}_x:\bm{\eta}(\bm{\theta}_x)\succeq\bm{\eta}(\bm{\theta}_y),
\end{equation}
where $\bm{\eta}(\bm{\theta}_z)$ is the $n$-dimensional performance vector associated with scheme $z$ and parameters $\bm{\theta}_z$. The definition of strict dominance between schemes is analogous.
\end{definition}

If we only consider the \gls{aoi}, without including state information in the optimization, we can envision a periodic policy in which the \gls{bs} sens a new update every $\Delta\in \mathbb{N^+}$ steps.
In this case, the optimal update interval $\Delta^*_{\text{AoI}}$ is then:
\begin{equation}
    \Delta^*_{\text{AoI}}(\beta) = \argmax_{\Delta \in \mathbb{N}^+} \sum_{s \in \mc{S}} \phi(s) \left(J(s \mid \Delta)-\sum_{n=0}^\infty\beta\gamma^{n\Delta}\right).
\end{equation}
In particular, we have that $\Delta(s)=\Delta^*_{\text{AoI}} \ \forall s\in\mc{S}$, i.e., the optimal \gls{aoi} $\Delta^*_{\text{AoI}}$ is the same for all states.

An adaptive pull-based strategy selecting the best $\Delta_{\text{pull}}(s)$ based on the state $s\in \mc{S}$ received in the last update can improve the reward function in \eqref{eq:tot_r} will respect the \gls{aoi}-based approach.
In this case, the optimization problem then becomes:
\begin{equation}
    \mb{\Delta}_{\text{pull}}^*(\beta) = \argmax_{\mb{\Delta}\in (\mathbb{N}^+)^{|\mc{S}|}} R_{\beta}(\bm{\Delta}).
\end{equation}

\begin{theorem}\label{th:aoi_pull}
The pull-based \gls{eff_com} strategy Pareto dominates the \gls{aoi}-based policy, i.e., $\bm{\Delta}^*_{\text{pull}}(\beta)\succeq\Delta^*_{\text{AoI}}(\beta)\,\forall\langle \set{S}, \set{A}, \mb{P}, r, \gamma,\beta \rangle$.
\end{theorem}
\begin{proof}
First, we can trivially show that a scheme working better than another for every value of $\beta$ is a stronger condition than Pareto dominance. We can then prove that this condition holds by \emph{reductio ad absurdum}: we consider a hypothetical optimal interval $\Delta^*_{\text{AoI}}$ which performs better than pull-based \gls{eff_com} for a given value of $\beta$. In this case, $\bm{\Delta}^*_{\text{pull}}$ cannot be optimal, as the vector in which all elements are equal to $\Delta^*_{\text{AoI}}$ is one of the possible choices for pull-based \gls{eff_com}.
\end{proof}

We observe that the pull-based strategy can always fall back to the same point as the \gls{aoi}-based one, simply by considering the same interval for every state.
In this case, the \gls{voi} is implicitly determined by the strategy: the Actuator only requests an update if the expected increase in the long-term reward is larger than the communication cost $\beta$, and that increase is precisely the value of the update.

We can further improve performance by adopting a push-based approach, allowing the \gls{bs} may be able to independently observe the system state $s_t$ and consequently decide whether to update the Actuator's knowledge.
This prevents us from defining a fixed state-dependent update period, as the policy $\pi_{B}(s,k,s_t)$ depends on the current state $s_t$ as well as on the last state update $s_{t-k}$ and the elapsed time $k$.
Particularly, the value function for the \gls{bs} is:
\begin{equation}
\begin{gathered}
    V_{B}(s,k,s')=\pi_{B}(s,k,s')(V_{B}(s',0,s')\!-\!\beta)+(1\!-\!\pi_{B}(s,k,s'))\\
    \times\sum_{s''\in\mc{S}}P^{\pi_{A}(s,k)}_{s's''}\left(r_{s',s''}^{\pi_{A}(s,k)}+\gamma V_{B}(s,k+1,s'')\right),
\end{gathered}
\end{equation}
and the optimal policy is the one maximizing the communication value.

The discovery of the optimal policies through \gls{api} can be seen as the solution to a Markov game~\cite{wang2002reinforcement}: the two agents act as players in a game where the moves are the possible policies and the payoff for each player is the expected reward in the initial state.
\begin{theorem}\label{th:nash}
The \gls{api} approach leads to a \gls{ne} policy in the push-based problem.
\end{theorem}
\begin{proof}
Firstly, we can trivially prove that the considered Markov game is an exact potential game~\cite{monderer1996potential}: as the reward for the two agents is the same, the expected long-term reward is a potential function for the game.
We then consider the \gls{api} strategy: each round of the iterated algorithm leads to the optimal policy when the strategy of the other agent is given, due to the optimality of standard \gls{pi}. The \gls{api} algorithm is then an \gls{ibr} scheme for the game, which leads to a \gls{ne} in a finite number of steps in all finite potential games~\cite{monderer1996potential}.
\end{proof}

However, reaching an \gls{ne} is \emph{not} a guarantee of Pareto optimality: games may have multiple \glspl{ne}, and finding the optimal one is PPAD-hard~\cite{deng2023complexity}. The push-based approach may be actively harmful, even with respect to an \gls{aoi} policy.

\begin{theorem}\label{th:push_pull}
The optimal \gls{eff_com} solution to the push-based problem, $\bm{\pi}^*_{B,\text{push}}$, Pareto dominates the pull-based \gls{eff_com} strategy:
\begin{equation}
 \bm{\pi}^*_{B,\text{push}}(\beta)\succeq\bm{\Delta}^*_{\text{pull}}(\beta)\succeq\Delta^*_{\text{AoI}}(\beta)\ \forall\langle \set{S}, \set{A}, \mb{P}, r, \gamma,\beta \rangle.
\end{equation}
However, the solution obtained by the \gls{api} strategy, $\bm{\pi}^{\text{API}}_{B,\text{push}}$, does not Pareto dominate the \gls{aoi}-based strategy:
\begin{equation}
\exists\langle \set{S}, \set{A}, \mb{P}, r, \gamma,\beta \rangle: \bm{\pi}^{\text{API}}_{B,\text{push}}(\beta)\nsucceq\Delta^*_{\text{AoI}}(\beta).
\end{equation}
\end{theorem}
\begin{proof}
 The proof of the first statement is simple, and follows the proof of Theorem~\ref{th:aoi_pull}: as the \gls{bs} can act with full knowledge of the state, the pull-based solution is a possible solution to the push-based problem, and in some cases, better solutions exist. The knowledge of the realization of the state can improve the reward, either by cutting unnecessary transmissions or by improving the actuator's performance.

As finding the optimal solution to a Markov game is PPAD-hard~\cite{deng2023complexity}, no polynomial-time algorithm can reliably find $\bm{\pi}^*_{B,\text{push}}$. We can then give a counterexample to prove the second part of the theorem: we consider a simple \gls{mdp} with $5$ states and $2$ actions, whose evolution is depicted in Fig.~\ref{fig:markov_model}.
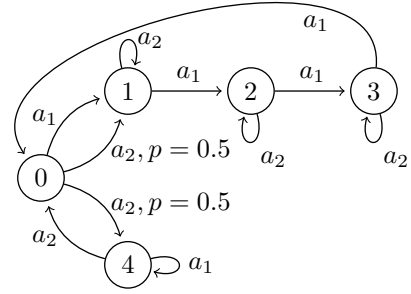
\begin{figure}
    \centering
    \input{markov_model}
    \caption{Example of a Markov model with $5$ states and $2$ actions.}
    \label{fig:markov_model}
\end{figure}
The two transition matrices corresponding to $a_1$ and $a_2$ are:
\begin{equation}
\mb{P}^{a_1} = {
\footnotesize\begin{bmatrix}
        0 & 1 & 0 & 0 & 0 \\
        0 & 0 & 1 & 0 & 0 \\
        0 & 0 & 0 & 1 & 0 \\
        1 & 0 & 0 & 0 & 0 \\
        0 & 0 & 0 & 0 & 1
    \end{bmatrix}},
\mb{P}^{a_2} ={
\footnotesize\begin{bmatrix}
        0 & 0.5 & 0 & 0 & 0.5 \\
        0 & 1 & 0 & 0 & 0 \\
        0 & 0 & 1 & 0 & 0 \\
        0 & 0 & 0 & 1 & 0 \\
        1 & 0 & 0 & 0 & 0
    \end{bmatrix}}.
\end{equation}
The reward is then always 0, except for when the environment transits to state 0, i.e., we have $r_{s,0}^a=1$ and $r_{s,s'}^a=0\ \forall s'\neq0$.

We can easily see that taking action $a_1$ in state $0$ leads to a loop with $4$ states, while taking action $a_2$ may lead to a shorter path back to the reward-giving state.
We consider a policy $\Delta_{\text{AoI}}$ that transmits in each odd step, ensuring that the actuator always knows if it lands in state $1$ or state $4$ after taking action $a_2$.
Its expected long-term reward is:
\begin{equation}
\begin{aligned}
 R(\Delta_{\text{AoI}})=&1-\gamma\beta+\frac{\gamma^2 R(\Delta_{\text{AoI}})}{2}-\frac{\gamma^3\beta}{2}+\frac{\gamma^4 R(\Delta_{\text{AoI}})}{2}\\
 =&\frac{2-(2\gamma+\gamma^3)\beta}{2(1-\gamma^2-\gamma^4)}.
\end{aligned}
 \end{equation}

Considering a push-based approach and applying the \gls{api} strategy, the final results depend on the initial policy of the \gls{bs}.
If the process starts from a policy that communicates often, e.g., one that always communicates the state, the algorithm will converge to the optimum joint policy, which only communicates if it deviates from the short cycle (i.e., if the system ends up in state 1).
In this case, the expected reward is
\begin{equation}
\begin{aligned}
 R(\bm{\pi}^*_{B,\text{push}})=&1-\gamma\beta+\frac{\gamma^2 R(\bm{\pi}^*_{B,\text{push}})}{2}+\frac{\gamma^4 R(\bm{\pi}^*_{B,\text{push}})}{2}\\
 =&\frac{2-\gamma\beta}{2(1-\gamma^2-\gamma^4)},
\end{aligned}
\end{equation}
which is better than the \gls{aoi} policy for any value of $\beta\in\mathbb{R}^+$ and $\gamma\in(0,1)$.

Instead, if the \gls{api} strategy starts from a policy that \emph{never} communicates, the actuator will take the conservative choice, and always take action $a_1$.
This is another \gls{ne} of the system, as the \gls{bs} should never communicate if the actuator's policy is independent of the state.
The reward for this solution is:
\begin{equation}
\begin{aligned}
 R(\bm{\pi}^{\text{API}}_{B,\text{push}})=&1+\gamma^4 R(\bm{\pi}^{\text{API}}_{B,\text{push}})=(1-\gamma^4)^{-1}.
\end{aligned}
\end{equation}
In this case, the \gls{api} solution is not Pareto dominant, as it performs worse than a simple \gls{aoi}-based strategy under the following conditions:
\begin{equation}
 \beta<\frac{2}{(2+\gamma^2)(1-\gamma^4)}.
\end{equation}
\end{proof}

The optimal policy might not be easy to obtain in more complex systems, and if we consider previously unknown cases in which \gls{pi} must be replaced by \gls{rl}, running multiple training procedures to achieve the optimal \gls{ne} might be expensive or impossible.
Despite this evidence, most of the literature on \gls{voi} has so far focused on \emph{push-based} solutions~\cite{gunduz2023timely}, which fall prey to this coordination problem.

\section{Conclusion}
\label{sec:conclusion}

This work analyzes \gls{eff_com} strategies in \glspl{cps}, considering the performance of the control system as the \gls{voi} and analyzing the interactions between the communication and control policies. We propose an analytical framework to adapt classical optimization tools such as \gls{pi} to this context, providing numerical results that show a strong dependency between the \gls{voi} and the structure of the underlying \gls{mdp}. Less predictable \glspl{mdp} and sparser rewards tend to lead to a higher value of updates, although there are interesting patterns in the interaction between communication and control.

Finally, our analysis revealed that the common push-based view of communication and control problems may not always be optimal, as the game theoretical properties of multiagent scenarios do not guarantee that a solution may be better even than a simple \gls{aoi}-based optimization. Future extensions of this work will concentrate on this conundrum, analyzing more complex examples of remote \glspl{pomdp} and devising heuristic strategies that provide good performance in practical scenarios.
\bibliographystyle{IEEEtran}
\bibliography{biblio.bib}

\end{document}

%% file: acronyms.tex
\newacronym{iot}{IoT}{Internet of Things}
\newacronym{ml}{ML}{machine learning}
\newacronym{dl}{DL}{Deep Learning}
\newacronym{marl}{MARL}{multi-agent reinforcement learning}
\newacronym{rl}{RL}{reinforcement learning}
\newacronym{decpomdp}{Dec-POMDP}{Decentralized Partially Observable Markov Decision Process }
\newacronym{acnomdp}{ACNO-MDP}{Action-Contingent Noiselessly Observable MDP}
\newacronym{pomdp}{POMDP}{Partially Observable Markov Decision Process}
\newacronym{lidar}{LIDAR}{Laser Imaging Detection and Ranging}
\newacronym{uav}{UAV}{Unmanned Aerial Vehicle}
\newacronym{dqn}{DQN}{Deep Q-Network}
\newacronym{dnn}{DNN}{deep neural network}
\newacronym{dial}{DIAL}{Differentiable Inter-Agent Learning}
\newacronym{api}{API}{Alternate Policy Iteration}
\newacronym{mdp}{MDP}{Markov Decision Process}
\newacronym{fov}{FoV}{Field of View}
\newacronym{cnn}{CNN}{Convolutional Neural Network}
\newacronym{nn}{NN}{neural network}
\newacronym{ddql}{DDQL}{Distributed Deep Q-Learning}
\newacronym{pdf}{PDF}{Probability Density Function}
\newacronym{ndpomdp}{ND-POMDP}{Networked Distributed Partially Observable Markov Decision Process}
\newacronym{ncs}{NCS}{Networked Control System}
\newacronym{cps}{CPS}{Cyber-Physical System}
\newacronym{radam}{RAdam}{Rectified Adam}
\newacronym{cdf}{CDF}{cumulative distribution function}
\newacronym{mpc}{MPC}{Model Predictive Control}
\newacronym{rv}{rv}{Random Variable}
\newacronym{qoe}{QoE}{Quality of Experience}
\newacronym{tlc}{TLC}{Telecommunications}
\newacronym{cml}{CML}{communications for machine learning}
\newacronym{mlc}{MLC}{machine learning for communications}
\newacronym{drl}{DRL}{Deep Reinforcement Learning}
\newacronym{rf}{RF}{Radio Frequency}
\newacronym{urllc}{URLLC}{Ultra-Reliable and Low-Latency Communications}
\newacronym{fl}{FL}{federated learning}
\newacronym{kpi}{KPI}{Key Performance Indicators}
\newacronym{mec}{MEC}{Mobile Edge Computing}
\newacronym{ei}{EI}{Edge Intelligence}
\newacronym{bs}{BS}{base station}
\newacronym{sdn}{SDN}{Software Defined Networking}
\newacronym{mimo}{MIMO}{Multiple-Input Multiple-Output}
\newacronym{gp}{GP}{Gaussian Process}
\newacronym{iiot}{IIoT}{Industrial Internet of Things}
\newacronym{csi}{CSI}{Channel State Information}
\newacronym{sgd}{SGD}{Stochastic Gradient Descent}
\newacronym{iid}{i.i.d.}{independent and identically distributed}
\newacronym{ofdm}{OFDM}{Orthogonal Frequency Division Multiplexing}
\newacronym{los}{LOS}{Line-of-Sight}
\newacronym{nlos}{NLOS}{Non-Line-of-Sight}
\newacronym{snr}{SNR}{Signal to Noise Ratio}
\newacronym{rb}{RB}{Resource Block}
\newacronym{6g}{6G}{sixth generation}
\newacronym{ai}{AI}{artificial intelligence}
\newacronym{sfl}{SFL}{Synchronous Federated Learning}
\newacronym{frfl}{FRFL}{Fixed Rate Federated Learning}
\newacronym{pgm}{PGM}{Probabilistic Graphical Model}
\newacronym{hmm}{HMM}{Hidden Markov Model}
\newacronym{elbo}{ELBO}{Evidence Lower Bound}
\newacronym{pmf}{PMF}{Probability Mass Function}
\newacronym{smab}{SMAB}{Stochastic Multi-Armed Bandit}
\newacronym{mab}{MAB}{multi-armed bandit}
\newacronym{mc}{MC}{Monte Carlo}
\newacronym{is}{IS}{Importance Sampling}
\newacronym{dms}{DMS}{discrete memoryless source}
\newacronym{ucb}{UCB}{upper confidence bound}
\newacronym{ser}{SER}{Symbol Error Rate}
\newacronym{sc}{SC}{Semantic Communications}
\newacronym{voi}{VoI}{Value of Information}
\newacronym{nlp}{NLP}{natural language processing}
\newacronym{ts}{TS}{Thompson Sampling}
\newacronym{cmab}{CMAB}{contextual multi-armed bandit}
\newacronym{rccmab}{RC-CMAB}{rate-constrained \gls{cmab}}
\newacronym{rcmab}{R-CMAB}{remote \gls{cmab}}
\newacronym{ib}{IB}{information bottleneck}
\newacronym{merl}{MERL}{maximum entropy reinforcement learning}
\newacronym{pac}{PAC}{Probably Approximately Correct}
\newacronym{aoi}{AoI}{Age of Information}
\newacronym{aoii}{AoII}{Age of Incorrect Information}
\newacronym{uoi}{UoI}{Urgency of Information}
\newacronym{vqvae}{VQ-VAE}{Vector Quantized Variational Autoencoder}
\newacronym{vae}{VAE}{Variational Autoencoder}
\newacronym{psnr}{PSNR}{Peak Signal to Noise Ratio}
\newacronym{mse}{MSE}{Mean Square Error}
\newacronym{lstm}{LSTM}{Long Short-Term Memory}
\newacronym{iomdp}{IOMDP}{Intermittently Observable Markov Decision Process}
\newacronym{pi}{PI}{Policy Iteration}
\newacronym{iql}{IQL}{Intermittent Q-Learning}
\newacronym{nc}{NC}{Network Core}
\newacronym{pql}{PQL}{Peekaboo Q-Learning}
\newacronym{eff_com}{EC}{Effective Communication}
\newacronym{ibr}{IBR}{Iterated Best Response}
\newacronym{ne}{NE}{Nash Equilibrium}

%% file: results/colormaps/reward_matrix_pull_sparse_reward.tex
\begin{tikzpicture}

\definecolor{darkgray176}{RGB}{176,176,176}

\begin{axis}[
height = \sfheight,
width = \sfwidth,
colorbar horizontal,
colorbar style={at={(0,1.25)},anchor=south west,height=0.4cm},
colormap/viridis,
point meta max=2,
point meta min=0,
tick align=outside,
tick pos=left,
x grid style={darkgray176},
xlabel={$\beta$},
xmin=-0.5, xmax=20.5,
xtick style={color=black},
xtick={0,4,8,12,16,20},
xticklabels={0,0.4,0.8,1.2,1.6,2},
y dir=reverse,
y grid style={darkgray176},
ylabel={$d$},
ymin=-0.5, ymax=13.5,
ytick style={color=black},
ytick={1.1,4.1,7.1,10.1,13.1},
yticklabels={0.9,0.7,0.5,0.3,0.1}
]
\addplot graphics [includegraphics cmd=\pgfimage,xmin=-0.5, xmax=20.5, ymin=13.5, ymax=-0.5] {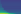};
\end{axis}

\end{tikzpicture}

%% file: results/colormaps/cost_matrix_pull_sparse_reward.tex
\begin{tikzpicture}

\definecolor{darkgray176}{RGB}{176,176,176}

\begin{axis}[
height = \sfheight,
width = \sfwidth,
colorbar horizontal,
colorbar style={at={(0,1.25)},anchor=south west,height=0.4cm},
colormap/viridis,
point meta max=1,
point meta min=0,
tick align=outside,
tick pos=left,
x grid style={darkgray176},
xlabel={$\beta$},
xmin=-0.5, xmax=20.5,
xtick style={color=black},
xtick={0,4,8,12,16,20},
xticklabels={0,0.4,0.8,1.2,1.6,2},
y dir=reverse,
y grid style={darkgray176},
ylabel={$d$},
ymin=-0.5, ymax=13.5,
ytick style={color=black},
ytick={1.1,4.1,7.1,10.1,13.1},
yticklabels={0.9,0.7,0.5,0.3,0.1}
]
\addplot graphics [includegraphics cmd=\pgfimage,xmin=-0.5, xmax=20.5, ymin=13.5, ymax=-0.5] {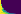};
\end{axis}

\end{tikzpicture}

%% file: results/colormaps/reward_matrix_pull_spread_reward.tex
\begin{tikzpicture}

\definecolor{darkgray176}{RGB}{176,176,176}

\begin{axis}[
height = \sfheight,
width = \sfwidth,
colorbar horizontal,
colorbar style={at={(0,1.25)},anchor=south west,height=0.4cm},
colormap/viridis,
point meta max=8,
point meta min=0,
tick align=outside,
tick pos=left,
x grid style={darkgray176},
xlabel={$\beta$},
xmin=-0.5, xmax=20.5,
xtick style={color=black},
xtick={0,4,8,12,16,20},
xticklabels={0,0.4,0.8,1.2,1.6,2},
y dir=reverse,
y grid style={darkgray176},
ylabel={$d$},
ymin=-0.5, ymax=13.5,
ytick style={color=black},
ytick={1.1,4.1,7.1,10.1,13.1},
yticklabels={0.9,0.7,0.5,0.3,0.1}
]
\addplot graphics [includegraphics cmd=\pgfimage,xmin=-0.5, xmax=20.5, ymin=13.5, ymax=-0.5] {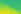};
\end{axis}

\end{tikzpicture}

%% file: results/colormaps/cost_matrix_pull_spread_reward.tex
\begin{tikzpicture}

\definecolor{darkgray176}{RGB}{176,176,176}

\begin{axis}[
height = \sfheight,
width = \sfwidth,
colorbar horizontal,
colorbar style={at={(0,1.25)},anchor=south west,height=0.4cm},
colormap/viridis,
point meta max=1,
point meta min=0,
tick align=outside,
tick pos=left,
x grid style={darkgray176},
xlabel={$\beta$},
xmin=-0.5, xmax=20.5,
xtick style={color=black},
xtick={0,4,8,12,16,20},
xticklabels={0,0.4,0.8,1.2,1.6,2},
y dir=reverse,
y grid style={darkgray176},
ylabel={$d$},
ymin=-0.5, ymax=13.5,
ytick style={color=black},
ytick={1.1,4.1,7.1,10.1,13.1},
yticklabels={0.9,0.7,0.5,0.3,0.1}
]
\addplot graphics [includegraphics cmd=\pgfimage,xmin=-0.5, xmax=20.5, ymin=13.5, ymax=-0.5] {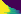};
\end{axis}

\end{tikzpicture}

%% file: results/colormaps/reward_matrix_push_sparse_reward.tex
\begin{tikzpicture}

\definecolor{darkgray176}{RGB}{176,176,176}

\begin{axis}[
height = \sfheight,
width = \sfwidth,
colorbar horizontal,
colorbar style={at={(0,1.25)},anchor=south west,height=0.4cm},
colormap/viridis,
point meta max=2,
point meta min=0,
tick align=outside,
tick pos=left,
x grid style={darkgray176},
xlabel={$\beta$},
xmin=-0.5, xmax=20.5,
xtick style={color=black},
xtick={0,4,8,12,16,20},
xticklabels={0,0.4,0.8,1.2,1.6,2},
y dir=reverse,
y grid style={darkgray176},
ylabel={$d$},
ymin=-0.5, ymax=13.5,
ytick style={color=black},
ytick={1.1,4.1,7.1,10.1,13.1},
yticklabels={0.9,0.7,0.5,0.3,0.1}
]
\addplot graphics [includegraphics cmd=\pgfimage,xmin=-0.5, xmax=20.5, ymin=13.5, ymax=-0.5] {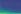};
\end{axis}

\end{tikzpicture}

%% file: results/colormaps/cost_matrix_push_sparse_reward.tex
\begin{tikzpicture}

\definecolor{darkgray176}{RGB}{176,176,176}

\begin{axis}[
height = \sfheight,
width = \sfwidth,
colorbar horizontal,
colorbar style={at={(0,1.25)},anchor=south west,height=0.4cm},
colormap/viridis,
point meta max=1,
point meta min=0,
tick align=outside,
tick pos=left,
x grid style={darkgray176},
xlabel={$\beta$},
xmin=-0.5, xmax=20.5,
xtick style={color=black},
xtick={0,4,8,12,16,20},
xticklabels={0,0.4,0.8,1.2,1.6,2},
y dir=reverse,
y grid style={darkgray176},
ylabel={$d$},
ymin=-0.5, ymax=13.5,
ytick style={color=black},
ytick={1.1,4.1,7.1,10.1,13.1},
yticklabels={0.9,0.7,0.5,0.3,0.1}
]
\addplot graphics [includegraphics cmd=\pgfimage,xmin=-0.5, xmax=20.5, ymin=13.5, ymax=-0.5] {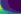};
\end{axis}

\end{tikzpicture}

%% file: results/colormaps/reward_matrix_push_spread_reward.tex
\begin{tikzpicture}

\definecolor{darkgray176}{RGB}{176,176,176}

\begin{axis}[
height = \sfheight,
width = \sfwidth,
colorbar horizontal,
colorbar style={at={(0,1.25)},anchor=south west,height=0.4cm},
colormap/viridis,
point meta max=8,
point meta min=0,
tick align=outside,
tick pos=left,
x grid style={darkgray176},
xlabel={$\beta$},
xmin=-0.5, xmax=20.5,
xtick style={color=black},
xtick={0,4,8,12,16,20},
xticklabels={0,0.4,0.8,1.2,1.6,2},
y dir=reverse,
y grid style={darkgray176},
ylabel={$d$},
ymin=-0.5, ymax=13.5,
ytick style={color=black},
ytick={1.1,4.1,7.1,10.1,13.1},
yticklabels={0.9,0.7,0.5,0.3,0.1}
]
\addplot graphics [includegraphics cmd=\pgfimage,xmin=-0.5, xmax=20.5, ymin=13.5, ymax=-0.5] {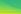};
\end{axis}

\end{tikzpicture}

%% file: results/colormaps/cost_matrix_push_spread_reward.tex
\begin{tikzpicture}

\definecolor{darkgray176}{RGB}{176,176,176}

\begin{axis}[
height = \sfheight,
width = \sfwidth,
colorbar horizontal,
colorbar style={at={(0,1.25)},anchor=south west,height=0.4cm},
colormap/viridis,
point meta max=1,
point meta min=0,
tick align=outside,
tick pos=left,
x grid style={darkgray176},
xlabel={$\beta$},
xmin=-0.5, xmax=20.5,
xtick style={color=black},
xtick={0,4,8,12,16,20},
xticklabels={0,0.4,0.8,1.2,1.6,2},
y dir=reverse,
y grid style={darkgray176},
ylabel={$d$},
ymin=-0.5, ymax=13.5,
ytick style={color=black},
ytick={1.1,4.1,7.1,10.1,13.1},
yticklabels={0.9,0.7,0.5,0.3,0.1}
]
\addplot graphics [includegraphics cmd=\pgfimage,xmin=-0.5, xmax=20.5, ymin=13.5, ymax=-0.5] {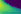};
\end{axis}

\end{tikzpicture}

%% file: results/AOI/all_states.tex
\begin{tikzpicture}

\definecolor{darkgray176}{RGB}{176,176,176}
\definecolor{darkorange25512714}{RGB}{255,127,14}
\definecolor{steelblue31119180}{RGB}{31,119,180}

\begin{axis}[
height = 0.21\linewidth,
width = 0.28\linewidth,
tick align=outside,
tick pos=left,
x grid style={darkgray176},
xlabel={Peak AoI},
xmin=-0.05, xmax=9.05,
xtick style={color=black},
y grid style={darkgray176},
ylabel={PMF},
xtick={0.5,2.5,4.5,6.5,8.5},
xticklabels={0,2,4,6,8},
ymin=0, ymax=0.6,
xmajorgrids,
ymajorgrids,
ytick style={color=black}
]
\draw[draw=none,fill=steelblue31119180,fill opacity=0.5] (axis cs:0,0) rectangle (axis cs:1,0);
\addlegendimage{ybar,ybar legend,draw=none,fill=steelblue31119180,fill opacity=0.5}
\addlegendentry{Pull}

\draw[draw=none,fill=steelblue31119180,fill opacity=0.5] (axis cs:1.2,0) rectangle (axis cs:1.45,0.0126968004062976);
\draw[draw=none,fill=steelblue31119180,fill opacity=0.5] (axis cs:2.2,0) rectangle (axis cs:2.45,0.436826364200666);
\draw[draw=none,fill=steelblue31119180,fill opacity=0.5] (axis cs:3.2,0) rectangle (axis cs:3.45,0.541335139100502);
\draw[draw=none,fill=steelblue31119180,fill opacity=0.5] (axis cs:4.2,0) rectangle (axis cs:4.45,0.00914169629253428);
\addlegendimage{ybar,ybar legend,draw=none,fill=darkorange25512714,fill opacity=0.5}
\addlegendentry{Push}

\draw[draw=none,fill=darkorange25512714,fill opacity=0.5] (axis cs:1.55,0) rectangle (axis cs:1.8,0.470063766153152);
\draw[draw=none,fill=darkorange25512714,fill opacity=0.5] (axis cs:2.55,0) rectangle (axis cs:2.8,0.0918119744935387);
\draw[draw=none,fill=darkorange25512714,fill opacity=0.5] (axis cs:3.55,0) rectangle (axis cs:3.8,0.150386547034592);
\draw[draw=none,fill=darkorange25512714,fill opacity=0.5] (axis cs:4.55,0) rectangle (axis cs:4.8,0.148129338073472);
\draw[draw=none,fill=darkorange25512714,fill opacity=0.5] (axis cs:5.55,0) rectangle (axis cs:5.8,0.0718356751876305);
\draw[draw=none,fill=darkorange25512714,fill opacity=0.5] (axis cs:6.55,0) rectangle (axis cs:6.8,0.0309237627673382);
\draw[draw=none,fill=darkorange25512714,fill opacity=0.5] (axis cs:7.55,0) rectangle (axis cs:7.8,0.0166469160882569);
\draw[draw=none,fill=darkorange25512714,fill opacity=0.5] (axis cs:8.55,0) rectangle (axis cs:8.8,0.00818238248405846);
\end{axis}

\end{tikzpicture}

%% file: results/AOI/heatmap_pull.tex
\begin{tikzpicture}

\definecolor{darkgray176}{RGB}{176,176,176}

\begin{axis}[
height = 0.21\linewidth,
width = 0.28\linewidth,
colorbar,
colorbar style={ylabel={}},
colormap/viridis,
point meta max=1,
point meta min=0,
tick align=outside,
tick pos=left,
x grid style={darkgray176},
xlabel={State},
xmin=-0.5, xmax=29.5,
xtick={0,9,19,29},
xticklabels={1,10,20,30},
xtick style={color=black},
y dir=reverse,
y grid style={darkgray176},
ylabel={Peak AoI},
ymin=0.5, ymax=9.5,
ytick style={color=black}
]
\addplot graphics [includegraphics cmd=\pgfimage,xmin=-0.5, xmax=29.5, ymin=9.5, ymax=-0.5] {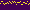};
\end{axis}

\end{tikzpicture}

%% file: results/AOI/heatmap_push.tex
\begin{tikzpicture}

\definecolor{darkgray176}{RGB}{176,176,176}

\begin{axis}[
height = 0.21\linewidth,
width = 0.28\linewidth,
colorbar,
colorbar style={ylabel={}},
colormap/viridis,
point meta max=1,
point meta min=0,
tick align=outside,
tick pos=left,
x grid style={darkgray176},
xtick={0,9,19,29},
xticklabels={1,10,20,30},
xlabel={State},
xmin=-0.5, xmax=29.5,
xtick style={color=black},
y dir=reverse,
y grid style={darkgray176},
ylabel={Peak AoI},
ymin=0.5, ymax=9.5,
ytick style={color=black}
]
\addplot graphics [includegraphics cmd=\pgfimage,xmin=-0.5, xmax=29.5, ymin=9.5, ymax=-0.5] {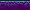};
\end{axis}

\end{tikzpicture}

%% file: markov_model.tex
\vspace{-0.7cm}
\begin{tikzpicture}[->, shorten >=2pt, line width=0.5 pt, node distance =1 cm]
                        ]
\node (n0)  [circle,draw] {$0$};
\node (n4)  [circle,draw,below right=of n0] {$4$};
\node (n1)  [circle,draw,above right=of n0] {$1$};
\node (n2)  [circle,draw, right=of n1] {$2$};
\node (n3)  [circle,draw, right=of n2] {$3$};

\path (n0) edge [bend left] node [right]{$a_2,p=0.5$} (n4);
\path (n0) edge [bend right] node [right]{$a_2,p=0.5$} (n1);
\path (n0) edge [bend left] node [left]{$a_1$} (n1);
\path (n1) edge node [above]{$a_1$} (n2);
\path (n2) edge node [above]{$a_1$} (n3);
\path (n3) edge [bend right=110] node [below,near start]{$a_1$} (n0);
\path (n4) edge [loop right] node [right]{$a_1$} (n4);
\path (n4) edge [bend left] node [left]{$a_2$} (n0);

\path (n2) edge [loop below] node [below right]{$a_2$} (n2);
\path (n3) edge [loop below] node [below right]{$a_2$} (n3);
\path (n1) edge [loop above] node [right]{$a_2$} (n1);

\end{tikzpicture}

%% file: main.bbl
\begin{thebibliography}{10}
\providecommand{\url}[1]{#1}
\csname url@samestyle\endcsname
\providecommand{\newblock}{\relax}
\providecommand{\bibinfo}[2]{#2}
\providecommand{\BIBentrySTDinterwordspacing}{\spaceskip=0pt\relax}
\providecommand{\BIBentryALTinterwordstretchfactor}{4}
\providecommand{\BIBentryALTinterwordspacing}{\spaceskip=\fontdimen2\font plus
\BIBentryALTinterwordstretchfactor\fontdimen3\font minus
  \fontdimen4\font\relax}
\providecommand{\BIBforeignlanguage}[2]{{%
\expandafter\ifx\csname l@#1\endcsname\relax
\typeout{** WARNING: IEEEtran.bst: No hyphenation pattern has been}%
\typeout{** loaded for the language `#1'. Using the pattern for}%
\typeout{** the default language instead.}%
\else
\language=\csname l@#1\endcsname
\fi
#2}}
\providecommand{\BIBdecl}{\relax}
\BIBdecl

\bibitem{zhang2019networked}
X.-M. Zhang, Q.-L. Han, X.~Ge, D.~Ding, L.~Ding, D.~Yue, and C.~Peng,
  ``Networked control systems: A survey of trends and techniques,''
  \emph{IEEE/CAA Journal of Automatica Sinica}, vol.~7, no.~1, pp. 1--17, 2019.

\bibitem{da2014internet}
L.~Da~Xu, W.~He, and S.~Li, ``Internet of things in industries: A survey,''
  \emph{IEEE Transactions on industrial informatics}, vol.~10, no.~4, pp.
  2233--2243, 2014.

\bibitem{mason2023multiagent}
F.~Mason, F.~Chiariotti, A.~Zanella, and P.~Popovski, ``Multi-agent
  reinforcement learning for pragmatic communication and control,'' 2023.

\bibitem{sutton2018reinforcement}
R.~S. Sutton and A.~G. Barto, \emph{Reinforcement learning: An
  introduction}.\hskip 1em plus 0.5em minus 0.4em\relax MIT press, 2018.

\bibitem{busoniu2008comprehensive}
L.~Busoniu, R.~Babuska, and B.~De~Schutter, ``A comprehensive survey of
  multiagent reinforcement learning,'' \emph{IEEE Transactions on Systems, Man,
  and Cybernetics, Part C (Applications and Reviews)}, vol.~38, no.~2, pp.
  156--172, 2008.

\bibitem{monahan1982state}
G.~E. Monahan, ``State of the art—a survey of partially observable {M}arkov
  decision processes: theory, models, and algorithms,'' \emph{Management
  science}, vol.~28, no.~1, pp. 1--16, 1982.

\bibitem{nam2021reinforcement}
H.~A. Nam, S.~Fleming, and E.~Brunskill, ``Reinforcement learning with state
  observation costs in action-contingent noiselessly observable {M}arkov
  decision processes,'' \emph{Advances in Neural Information Processing
  Systems}, vol.~34, pp. 15\,650--15\,666, 2021.

\bibitem{pareto1919manuale}
V.~Pareto, \emph{Manuale di economia politica con una introduzione alla scienza
  sociale}.\hskip 1em plus 0.5em minus 0.4em\relax Società Editrice Libraia,
  1919.

\bibitem{rostami2019wake}
S.~Rostami, S.~Lagen, M.~Costa, M.~Valkama, and P.~Dini, ``Wake-up radio based
  access in {5G} under delay constraints: Modeling and optimization,''
  \emph{IEEE Transactions on Communications}, vol.~68, no.~2, pp. 1044--1057,
  2019.

\bibitem{pineau2003point}
J.~Pineau, G.~Gordon, S.~Thrun \emph{et~al.}, ``Point-based value iteration: An
  anytime algorithm for pomdps,'' in \emph{Ijcai}, vol.~3, 2003, pp.
  1025--1032.

\bibitem{smith2005point}
T.~Smith and R.~Simmons, ``Point-based pomdp algorithms: improved analysis and
  implementation,'' in \emph{Proceedings of the Twenty-First Conference on
  Uncertainty in Artificial Intelligence}, 2005, pp. 542--549.

\bibitem{tung2021effective}
T.-Y. Tung, S.~Kobus, J.~P. Roig, and D.~G{\"u}nd{\"u}z, ``Effective
  communications: A joint learning and communication framework for multi-agent
  reinforcement learning over noisy channels,'' \emph{IEEE Journal on Selected
  Areas in Communications}, vol.~39, no.~8, pp. 2590--2603, 2021.

\bibitem{wang2002reinforcement}
X.~Wang and T.~Sandholm, ``Reinforcement learning to play an optimal nash
  equilibrium in team markov games,'' \emph{Advances in Neural Information
  Processing Systems (NIPS)}, vol.~15, 2002.

\bibitem{monderer1996potential}
D.~Monderer and L.~S. Shapley, ``Potential games,'' \emph{Games and Economic
  Behavior}, vol.~14, no.~1, pp. 124--143, 1996.

\bibitem{deng2023complexity}
X.~Deng, N.~Li, D.~Mguni, J.~Wang, and Y.~Yang, ``On the complexity of
  computing {Markov} perfect equilibrium in general-sum stochastic games,''
  \emph{National Science Review}, vol.~10, no.~1, p. nwac256, 2023.

\bibitem{gunduz2023timely}
D.~G{\"u}nd{\"u}z, F.~Chiariotti, K.~Huang, A.~E. Kal{\o}r, S.~Kobus, and
  P.~Popovski, ``Timely and massive communication in {6G}: Pragmatics,
  learning, and inference,'' \emph{IEEE BITS the Information Theory Magazine},
  2023.

\end{thebibliography}
